\documentclass[11pt,a4paper,twocolumn]{article}
\usepackage{epsfig}
\usepackage{lscape}
\usepackage{amsmath}
\usepackage{amssymb}
\usepackage{amsthm}
\usepackage{amsfonts}
\usepackage{color}
\usepackage{rotating}

\usepackage{latexsym}           %Fonts
\usepackage{graphicx}
\usepackage{ifpdf}
\usepackage{enumerate}          %Enumeration Style
\usepackage{mathrsfs}           %Math Fonts
\usepackage[plain]{algorithm}
\usepackage{algorithmic}
\usepackage{makeidx}            %index
\usepackage{array}
\usepackage{longtable}
\usepackage{geometry}           %fuer Layout
\usepackage{verbatim}           %Textdatei wird dargestellt (z.B. Algorithmus)

\usepackage{setspace,vmargin}

\begin{document}

%\setmarginsrb{.6cm}{.6cm}{.6cm}{1cm}{0pt}{0mm}{0pt}{.6in}
\setmarginsrb{1cm}{1cm}{1cm}{.6in}{0pt}{0mm}{0pt}{.4in}

\newtheorem{theo}{Theorem} %[section]
\newtheorem{defi}[theo]{Definition}
\newtheorem{coro}[theo]{Corollary}
\newtheorem{prop}[theo]{Proposition}
\newtheorem{lem}[theo]{Lemma}
\newtheorem{rem}[theo]{Remark}
\newtheorem{conj}[theo]{Conjecture}

\newcommand{\E}{\mathbb{E}}     % E(X) -- expectation
\renewcommand{\P}{\mathbb{P}}   % P(X=1) -- probability
\newcommand{\D}{\displaystyle{}}
\newcommand{\T}{\textstyle{}}

\newcommand{\bE}{{\mathbb E}}
\newcommand{\bJ}{{\mathbb J}}
\newcommand{\bN}{{\mathbb N}}
\newcommand{\bP}{{\mathbb P}}
\newcommand{\bR}{{\mathbb R}}

\newcommand{\cA}{{ \mathcal A}}
\newcommand{\cB}{{ \mathcal B}}
\newcommand{\cC}{{\mathcal C}}
\newcommand{\cH}{{ \mathcal H}}
\newcommand{\cL}{{\mathcal L}}
\newcommand{\cN}{{ \mathcal N}}
\newcommand{\cP}{{\mathcal P}}
\newcommand{\cQ}{{ \mathcal Q}}
\newcommand{\cS}{{\mathcal S}}
\newcommand{\cT}{{ \mathcal T}}
\newcommand{\cV}{{ \mathcal V}}
\newcommand{\sA}{{\mathscr A}}

\newcommand{\sB}{{\mathscr B}}
\newcommand{\sC}{{\mathscr C}}
\newcommand{\sL}{{\mathscr L}}
\newcommand{\sS}{{\mathscr S}}
\newcommand{\sT}{{\mathscr G}}

\newcommand{\tanja}{\textcolor{black}}
\newcommand{\james}{\textcolor{black}}

%%%%%%%%%%%%%%%%%%%%%%%%%%%%%%%%%%%%%%%%%%%%%%%%%%%%%%%
%%%%%%%%%%%%%%%%%%%%%%%%%%%%%%%%%%%%%%%%%%%%%%%%%%%%%%%
%%%%%%%%%%%%%%%%%%%%%%%%%%%%%%%%%%%%%%%%%%%%%%%%%%%%%%%
%%%%%%%%%%%%%%%%%%%%%%%%%%%%%%%%%%%%%%%%%%%%%%%%%%%%%%%
%%%%%%%%%%%%%%%%%%%%%%%%%%%%%%%%%%%%%%%%%%%%%%%%%%%%%%%
%%%%%%%%%%%%%%%%%%%%%%%%%%%%%%%%%%%%%%%%%%%%%%%%%%%%%%%
%%%%%%%%%%%%%%%%%%%%%%%%%%%%%%%%%%%%%%%%%%%%%%%%%%%%%%%
%%%%%%%%%%%%%%%%%%%%%%%%%%%%%%%%%%%%%%%%%%%%%%%%%%%%%%%
%%%%%%%%%%%%%%%%%%%%%%%%%%%%%%%%%%%%%%%%%%%%%%%%%%%%%%%
%%%%%%%%%%%%%%%%%%%%%%%%%%%%%%%%%%%%%%%%%%%%%%%%%%%%%%%

\title{A polynomial time algorithm for calculating the probability of a ranked gene tree given a species tree}
\author{Tanja Stadler$^1$ \& James H.~Degnan$^{2,3}$\\ \noindent $^1$Institute of Integrative Biology
Universit\"atsstrasse 16,
8092, Z\"urich, Switzerland\\
 $^2$Dept.~of Mathematics and Statistics, Private Bag
4800, University of Canterbury\\
 Christchurch 8140 New Zealand\\ $^3$National Institute of Mathematical and Biological Synthesis, Knoxville, Tennessee, USA}

\maketitle

%\noindent $^1$Department of Human Genetics, University of Michigan,
%Ann Arbor, Michigan 48109 USA \\

%
%\noindent $^4$Center for Computational Medicine and Biology and the
%Life Sciences Institute, University of Michigan, Ann Arbor, Michigan
%48109 USA \\

%\begin{rem}
%Theorem 6 makes the whole thing $O(n^5)$. Everything else is $O(n^3)$. Maybe we can find a faster way for Theorem 6? I need to change Theorem 6, such that we use Wakeleys formula!
%new figure 1 - no color, no ranked history
%\end{rem}

%%%%%%%%%%%%

%\doublespacing

\begin{abstract}
In this paper, we provide a polynomial time algorithm to calculate the probability of a {\it ranked} gene tree topology for a given species tree, where a ranked tree topology is a tree topology with the internal vertices being ordered. The probability of a gene tree topology can thus be calculated in polynomial time if the number of orderings of the internal vertices is a polynomial number. However,  the complexity of calculating the probability of a gene tree topology with an exponential number of rankings for a given species tree remains unknown.
\end{abstract}

\section{Introduction}
Phylogenetic reconstruction methods aim to infer the species phylogeny which gave rise to a group of extant species. Typically, this species phylogeny is obtained based on genetic data from representative individuals of each extant species.
%This genetic data evolved based on replication of individuals, thus we reconstruct gene trees instead of species trees. 
The ancestries of genes at different loci form gene trees which do not necessarily have the same topology as the species tree. 
Gene tree topologies and species tree topologies might be different due to such phenomena as incomplete lineage sorting, gene duplication, recombination within gene loci, and horizontal gene transfer \cite{DegnanAndRosenberg09}.  In this paper, we focus on incomplete lineage sorting as the mechanism for incongruence of gene tree and species tree topologies, in which two gene lineages do not coalesce in the most recent population ancestral to the individuals from which the genes were sampled.  As an example, the lineages sampled from species $A$ and $B$ in Figure 1b do not coalesce until the population ancestral to species $A$, $B$, and $C$, thus allowing the $B$ and $C$ lineages in the gene tree to have a more recent common ancestor than lineages $A$ and $B$.

\begin{figure*}[!t]
\setlength{\unitlength}{0.05 mm}%  
  \begin{picture}(1372.0, 3204.0)(0,0)
  \put(20,1500){\includegraphics[width=.49\textwidth]{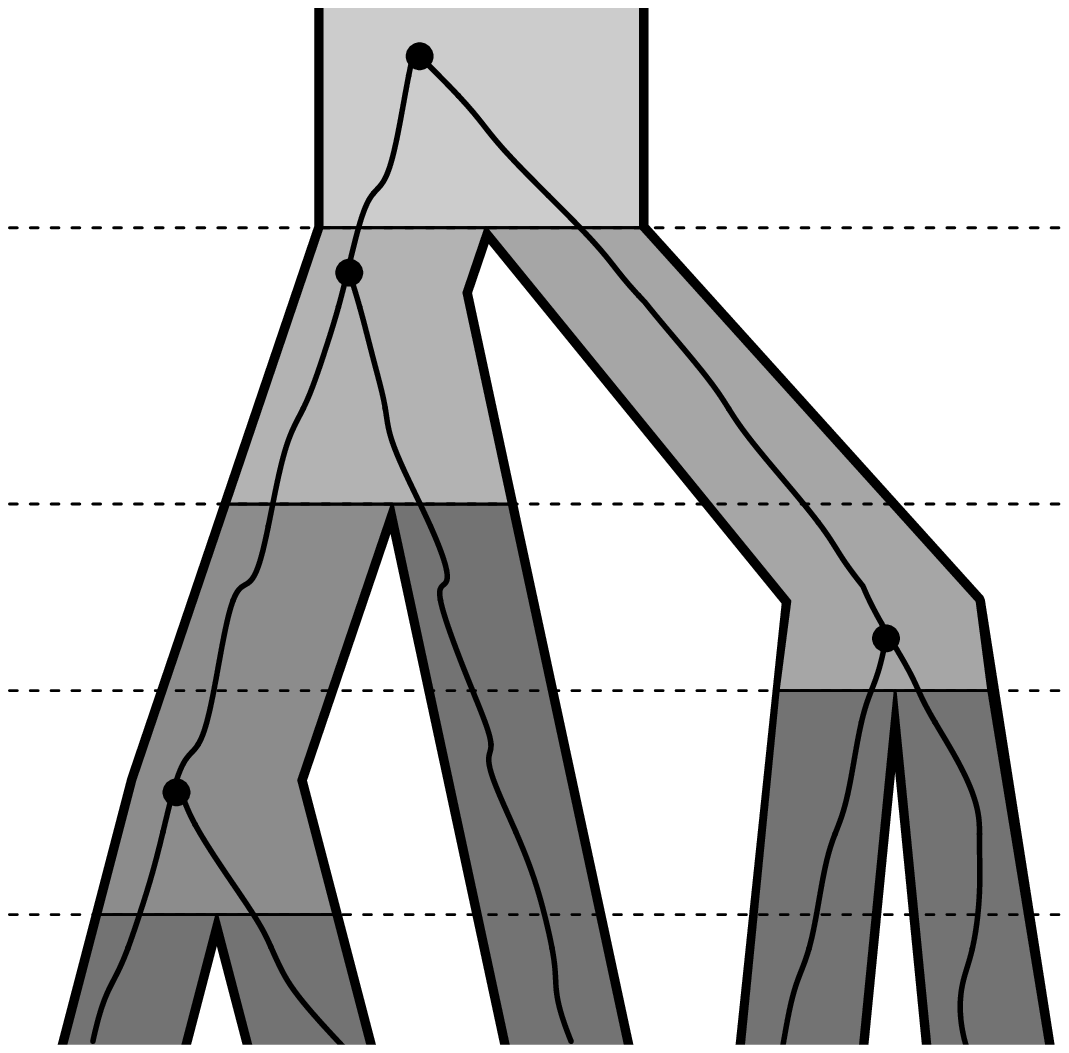}}
  \put(1800,1500){\includegraphics[width=.49\textwidth]{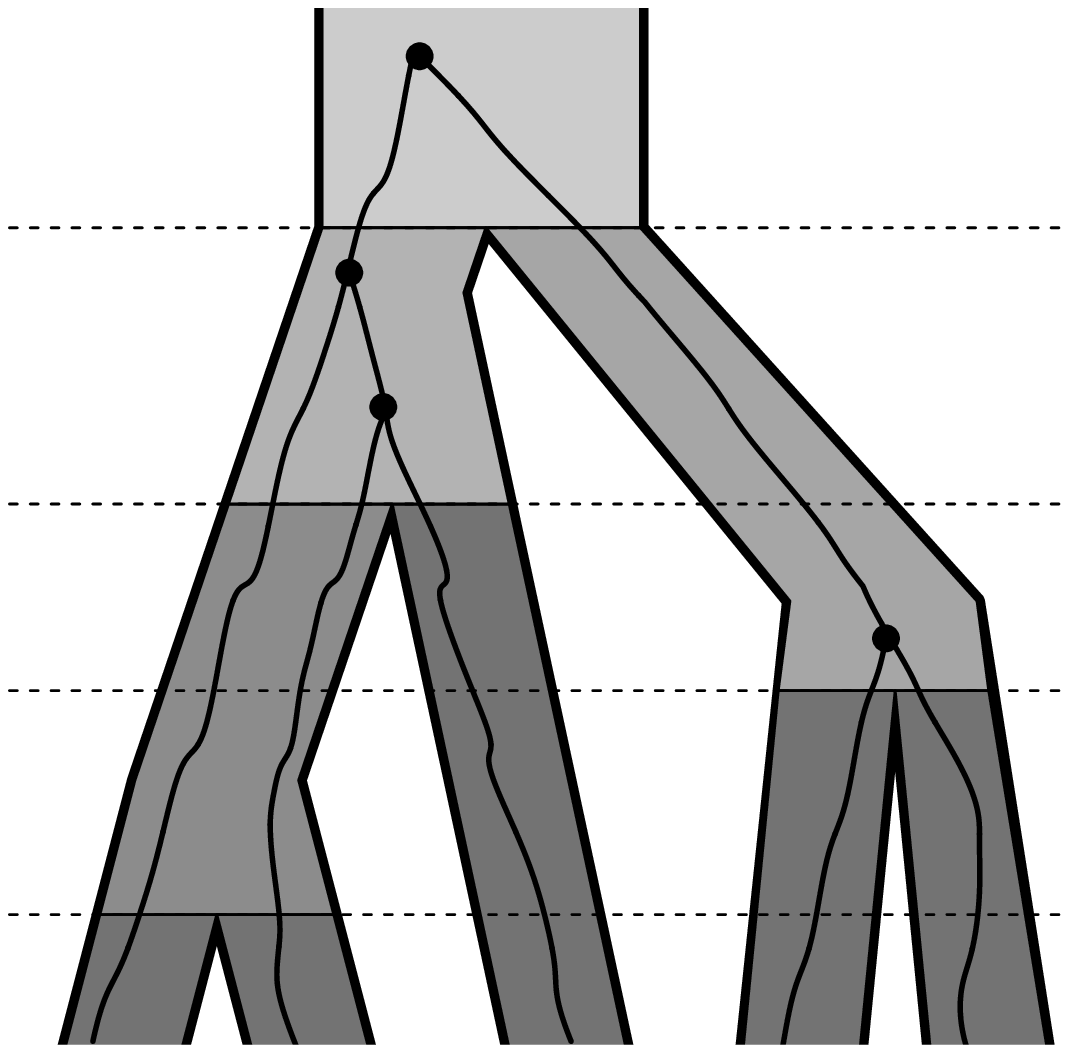}}
\put(0,0){\includegraphics[width=.49\textwidth]{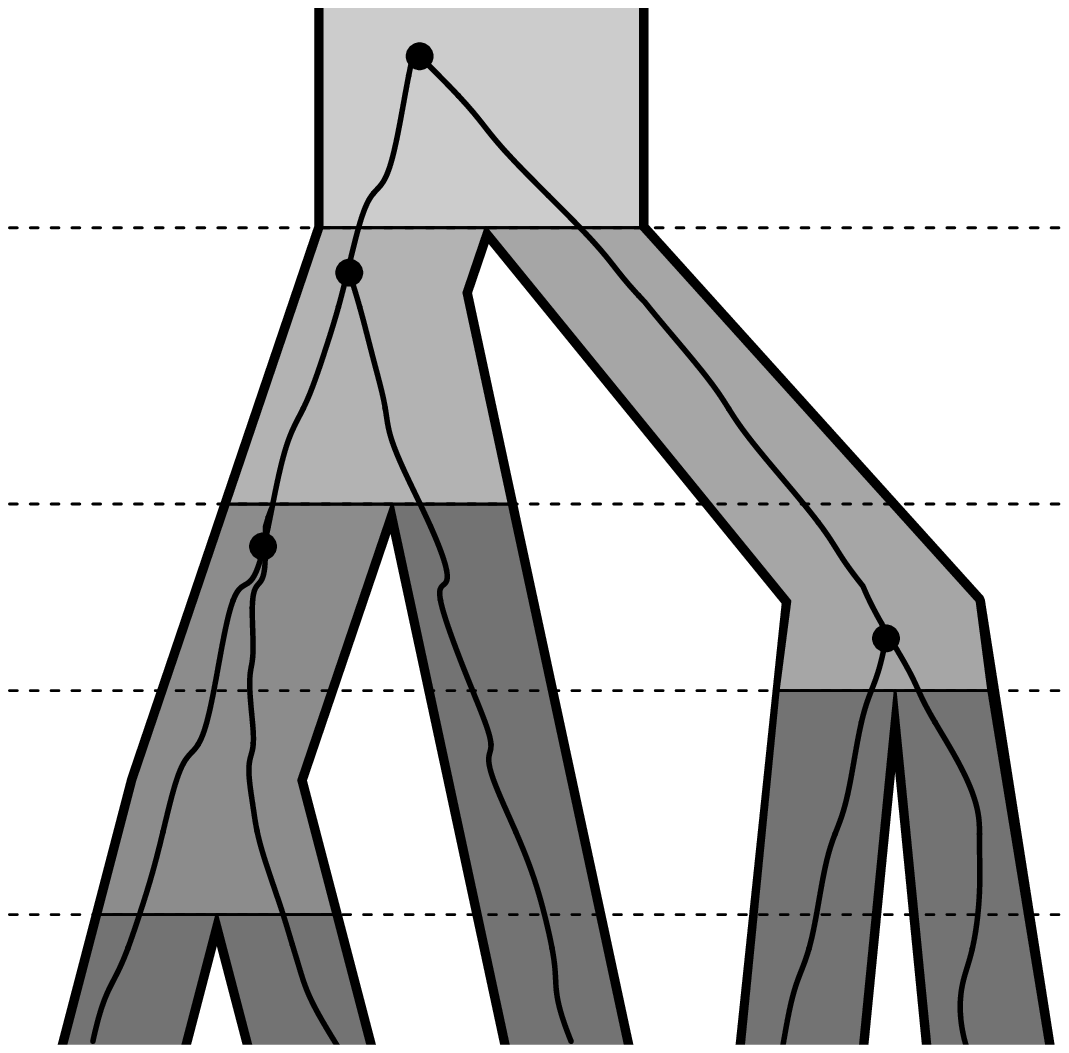}}
\put(1800,0){\includegraphics[width=.49\textwidth]{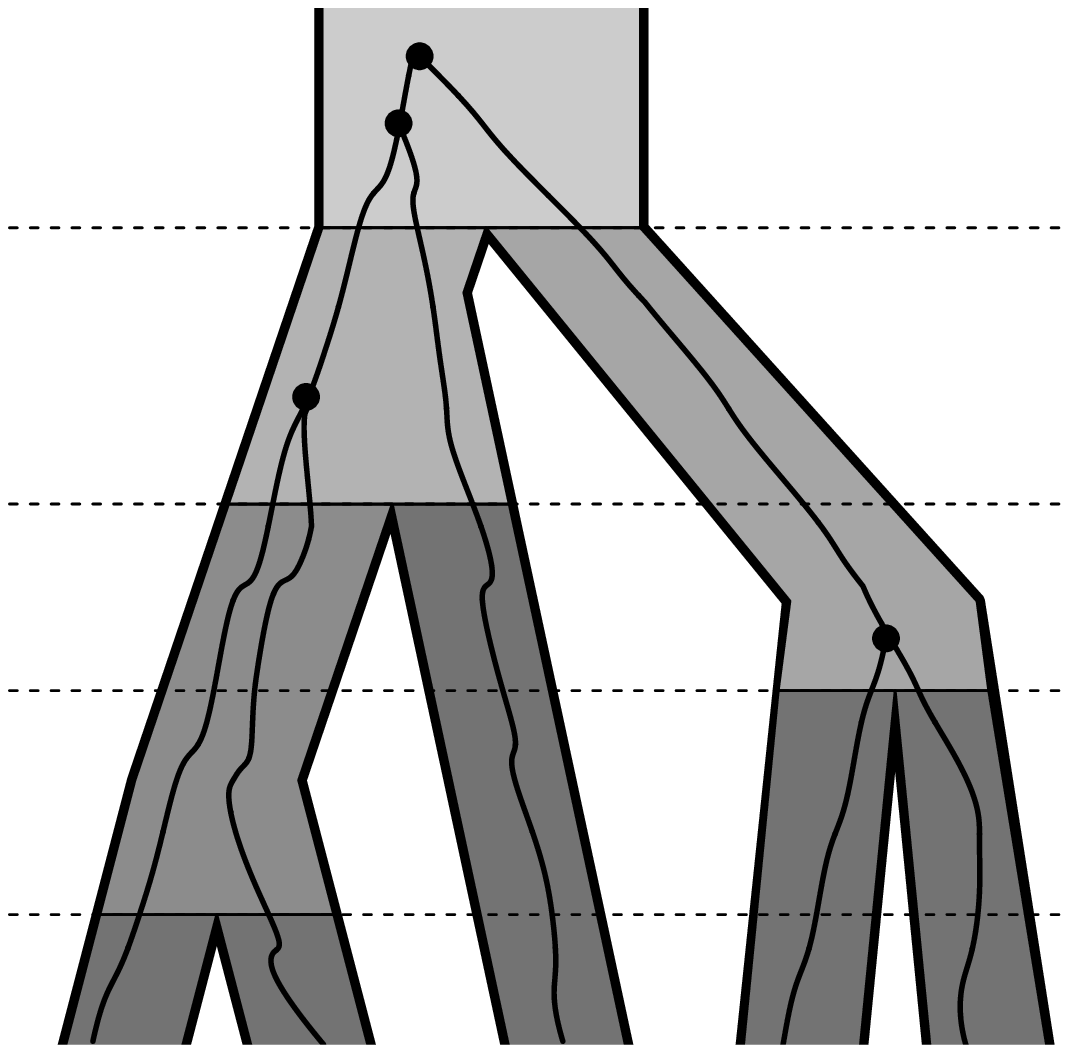}}
\put(420.00,290.00){\fontsize{11.23}{14.07}\selectfont   \makebox(72.0, 100.0)[c]{$A$\strut}}
\put(660.00,290.00){\fontsize{11.23}{14.07}\selectfont   \makebox(72.0, 100.0)[c]{$B$\strut}}
\put(910.00,290.00){\fontsize{11.23}{14.07}\selectfont   \makebox(72.0, 100.0)[c]{$C$\strut}}
\put(1150.00,290.00){\fontsize{11.23}{14.07}\selectfont   \makebox(72.0, 100.0)[c]{$D$\strut}}
\put(1350.00,290.00){\fontsize{11.23}{14.07}\selectfont  \makebox(72.0, 100.0)[c]{$E$\strut}}
\put(2200.00,290.00){\fontsize{11.23}{14.07}\selectfont   \makebox(72.0, 100.0)[c]{$A$\strut}}
\put(2420.00,290.00){\fontsize{11.23}{14.07}\selectfont   \makebox(72.0, 100.0)[c]{$B$\strut}}
\put(2690.00,290.00){\fontsize{11.23}{14.07}\selectfont   \makebox(72.0, 100.0)[c]{$C$\strut}}
\put(2950.00,290.00){\fontsize{11.23}{14.07}\selectfont   \makebox(72.0, 100.0)[c]{$D$\strut}}
\put(3150.00,290.00){\fontsize{11.23}{14.07}\selectfont  \makebox(72.0, 100.0)[c]{$E$\strut}}
\put(420.00,1790.00){\fontsize{11.23}{14.07}\selectfont   \makebox(72.0, 100.0)[c]{$A$\strut}}
\put(660.00,1790.00){\fontsize{11.23}{14.07}\selectfont   \makebox(72.0, 100.0)[c]{$B$\strut}}
\put(910.00,1790.00){\fontsize{11.23}{14.07}\selectfont   \makebox(72.0, 100.0)[c]{$C$\strut}}
\put(1150.00,1790.00){\fontsize{11.23}{14.07}\selectfont   \makebox(72.0, 100.0)[c]{$D$\strut}}
\put(1350.00,1790.00){\fontsize{11.23}{14.07}\selectfont  \makebox(72.0, 100.0)[c]{$E$\strut}}
\put(2200.00,1790.00){\fontsize{11.23}{14.07}\selectfont   \makebox(72.0, 100.0)[c]{$A$\strut}}
\put(2420.00,1790.00){\fontsize{11.23}{14.07}\selectfont   \makebox(72.0, 100.0)[c]{$B$\strut}}
\put(2690.00,1790.00){\fontsize{11.23}{14.07}\selectfont   \makebox(72.0, 100.0)[c]{$C$\strut}}
\put(2950.00,1790.00){\fontsize{11.23}{14.07}\selectfont   \makebox(72.0, 100.0)[c]{$D$\strut}}
\put(3150.00,1790.00){\fontsize{11.23}{14.07}\selectfont  \makebox(72.0, 100.0)[c]{$E$\strut}}
\put(270.00,1230.00){\fontsize{11.23}{14.07}\selectfont   \makebox(72.0, 100.0)[c]{$s_1$\strut}}
\put(270.00,940.00){\fontsize{11.23}{14.07}\selectfont   \makebox(72.0, 100.0)[c]{$s_2$\strut}}
\put(270.00,740.00){\fontsize{11.23}{14.07}\selectfont   \makebox(72.0, 100.0)[c]{$s_3$\strut}}
\put(270.00,500.00){\fontsize{11.23}{14.07}\selectfont   \makebox(72.0, 100.0)[c]{$s_4$\strut}}
\put(270.00,2730.00){\fontsize{11.23}{14.07}\selectfont   \makebox(72.0, 100.0)[c]{$s_1$\strut}}
\put(270.00,2440.00){\fontsize{11.23}{14.07}\selectfont   \makebox(72.0, 100.0)[c]{$s_2$\strut}}
\put(270.00,2240.00){\fontsize{11.23}{14.07}\selectfont   \makebox(72.0, 100.0)[c]{$s_3$\strut}}
\put(270.00,2000.00){\fontsize{11.23}{14.07}\selectfont   \makebox(72.0, 100.0)[c]{$s_4$\strut}}
\put(2050.00,2730.00){\fontsize{11.23}{14.07}\selectfont   \makebox(72.0, 100.0)[c]{$s_1$\strut}}
\put(2050.00,2440.00){\fontsize{11.23}{14.07}\selectfont   \makebox(72.0, 100.0)[c]{$s_2$\strut}}
\put(2050.00,2240.00){\fontsize{11.23}{14.07}\selectfont   \makebox(72.0, 100.0)[c]{$s_3$\strut}}
\put(2050.00,2000.00){\fontsize{11.23}{14.07}\selectfont   \makebox(72.0, 100.0)[c]{$s_4$\strut}}
\put(2050.00,1230.00){\fontsize{11.23}{14.07}\selectfont   \makebox(72.0, 100.0)[c]{$s_1$\strut}}
\put(2050.00,940.00){\fontsize{11.23}{14.07}\selectfont   \makebox(72.0, 100.0)[c]{$s_2$\strut}}
\put(2050.00,740.00){\fontsize{11.23}{14.07}\selectfont   \makebox(72.0, 100.0)[c]{$s_3$\strut}}
\put(2050.00,500.00){\fontsize{11.23}{14.07}\selectfont   \makebox(72.0, 100.0)[c]{$s_4$\strut}}
\put(840.00,1410.00){\fontsize{11.23}{14.07}\selectfont   \makebox(72.0, 100.0)[c]{$u_1$\strut}}
\put(840.00,2910.00){\fontsize{11.23}{14.07}\selectfont   \makebox(72.0, 100.0)[c]{$u_1$\strut}}
\put(640.00,890.00){\fontsize{11.23}{14.07}\selectfont    \makebox(72.0, 100.0)[c]{$u_3$\strut}}
\put(740.00,1180.00){\fontsize{11.23}{14.07}\selectfont    \makebox(72.0, 100.0)[c]{$u_2$\strut}}
\put(2620.00,2910.00){\fontsize{11.23}{14.07}\selectfont   \makebox(72.0, 100.0)[c]{$u_1$\strut}}
\put(2620.00,1410.00){\fontsize{11.23}{14.07}\selectfont   \makebox(72.0, 100.0)[c]{$u_1$\strut}}
\put(2570.00,2550.00){\fontsize{11.23}{14.07}\selectfont    \makebox(72.0, 100.0)[c]{$u_3$\strut}}
\put(580.00,2120.00){\fontsize{11.23}{14.07}\selectfont    \makebox(72.0, 100.0)[c]{$u_4$\strut}}
\put(2540.00,2680.00){\fontsize{11.23}{14.07}\selectfont    \makebox(72.0, 100.0)[c]{$u_2$\strut}}
\put(2480.00,1390.00){\fontsize{11.23}{14.07}\selectfont    \makebox(72.0, 100.0)[c]{$u_2$\strut}}
\put(2490.00,1060.00){\fontsize{11.23}{14.07}\selectfont    \makebox(72.0, 100.0)[c]{$u_3$\strut}}
\put(760.00,2680.00){\fontsize{11.23}{14.07}\selectfont    \makebox(72.0, 100.0)[c]{$u_2$\strut}}
\put(1320.00,2300.00){\fontsize{11.23}{14.07}\selectfont    \makebox(72.0, 100.0)[c]{$u_3$\strut}}
\put(1300.00,800.00){\fontsize{11.23}{14.07}\selectfont    \makebox(72.0, 100.0)[c]{$u_4$\strut}}
\put(3100.00,2300.00){\fontsize{11.23}{14.07}\selectfont    \makebox(72.0, 100.0)[c]{$u_4$\strut}}
\put(3100.00,800.00){\fontsize{11.23}{14.07}\selectfont    \makebox(72.0, 100.0)[c]{$u_4$\strut}}
\put(1480.00,1340.00){\fontsize{11.23}{14.07}\selectfont \makebox(72.0, 100.0)[c]{$\tau_1$\strut}}
\put(1480.00,1080.00){\fontsize{11.23}{14.07}\selectfont  \makebox(72.0, 100.0)[c]{$\tau_2$\strut}}
\put(1480.00,820.00){\fontsize{11.23}{14.07}\selectfont  \makebox(72.0, 100.0)[c]{$\tau_3$\strut}}
\put(1480.00,600.00){\fontsize{11.23}{14.07}\selectfont  \makebox(72.0, 100.0)[c]{$\tau_4$\strut}}
\put(1480.00,2840.00){\fontsize{11.23}{14.07}\selectfont \makebox(72.0, 100.0)[c]{$\tau_1$\strut}}
\put(1480.00,2580.00){\fontsize{11.23}{14.07}\selectfont  \makebox(72.0, 100.0)[c]{$\tau_2$\strut}}
\put(1480.00,2320.00){\fontsize{11.23}{14.07}\selectfont  \makebox(72.0, 100.0)[c]{$\tau_3$\strut}}
\put(1480.00,2100.00){\fontsize{11.23}{14.07}\selectfont  \makebox(72.0, 100.0)[c]{$\tau_4$\strut}}
\put(3260.00,2840.00){\fontsize{11.23}{14.07}\selectfont \makebox(72.0, 100.0)[c]{$\tau_1$\strut}}
\put(3260.00,2580.00){\fontsize{11.23}{14.07}\selectfont  \makebox(72.0, 100.0)[c]{$\tau_2$\strut}}
\put(3260.00,2320.00){\fontsize{11.23}{14.07}\selectfont  \makebox(72.0, 100.0)[c]{$\tau_3$\strut}}
\put(3260.00,2100.00){\fontsize{11.23}{14.07}\selectfont  \makebox(72.0, 100.0)[c]{$\tau_4$\strut}}
\put(3260.00,1340.00){\fontsize{11.23}{14.07}\selectfont \makebox(72.0, 100.0)[c]{$\tau_1$\strut}}
\put(3260.00,1080.00){\fontsize{11.23}{14.07}\selectfont  \makebox(72.0, 100.0)[c]{$\tau_2$\strut}}
\put(3260.00,820.00){\fontsize{11.23}{14.07}\selectfont  \makebox(72.0, 100.0)[c]{$\tau_3$\strut}}
\put(3260.00,600.00){\fontsize{11.23}{14.07}\selectfont  \makebox(72.0, 100.0)[c]{$\tau_4$\strut}}
\put(50.00,3000.00){\fontsize{14.23}{17.07}\selectfont   \makebox(72.0, 100.0)[c]{(a)\strut}}
\put(2050.00,3000.00){\fontsize{14.23}{17.07}\selectfont   \makebox(72.0, 100.0)[c]{(b)\strut}}
\put(50.00,1500.00){\fontsize{14.23}{17.07}\selectfont   \makebox(72.0, 100.0)[c]{(c)\strut}}
\put(2050.00,1500.00){\fontsize{14.23}{17.07}\selectfont   \makebox(72.0, 100.0)[c]{(d)\strut}}
%\put(600.00,1180.00){\fontsize{11.23}{14.07}\selectfont  \makebox(172.0,100.0)[c]{Species tree $\cT_{\text{LRL}}$, gene tree $\mathcal G_{\text{RLL}}$\strut}}
%\put(600.00,1280.00){\fontsize{11.23}{14.07}\selectfont  \makebox(172.0,100.0)[c]{Rank\strut}}

\end{picture}
\vspace{-1cm}
\caption{In (a)--(d) the ranked species tree topology is $(((A,B)_4,C)_2,(D,E)_3)_1$.  (a) The ranked gene tree matches the ranked species tree.  (b)  The (ranked or unranked) gene tree does not match the species tree, and there is an incomplete lineage sorting event (a deep coalescence) because the lineages from species $A$ and $B$ fail to coalesce more recently than $s_2$. (c) The gene tree and species tree have the same unranked topology but have different ranked topologies, as $D$ and $E$ coalesce in the gene tree more recently than $A$ and $B$, while $A$ and $B$ is the most recent divergence in the species tree. The gene tree in (c) has ranked topology $(((A,B)_3,C)_2,(D,E)_4)_1$.  In (c), there are no incomplete lineage sorting events (no deep coalescences); however, there is an extra lineage at time $s_3$ which leads to the gene tree and species tree having different rankings.  In (c), all coalescences occur in the most recent possible interval consistent with the ranked gene tree, and 
 we have $\ell_1=2, \ell_2=3, \ell_3=5, \ell_4=5$, and $g_1 =  2$, $g_2 = 3$, $g_3 = 5$, $g_4 = 5$. (d) A gene tree with the same ranked topology as the gene tree in (c) but with coalescences occurring in different intervals. }
  \label{FigNotation}
\end{figure*}

 Given a fixed species tree, and assuming the gene tree evolved under the multi-species coalescent \cite{DegnanAndRosenberg09},
the most probable gene tree topology can have a different topology from that of the species tree. Such a gene tree topology is called an anomalous gene tree. In fact, for every species tree topology with at least $5$ leaves, we can choose edge lengths in the species tree topology such that anomalous gene trees exist \cite{DegnanAndRosenberg06}. This implies that  the gene tree topology appearing most often when considering different genes might not agree with the  species tree topology, thus we cannot use a simple majority-heuristic to infer the species tree from a collection of gene trees. 
Instead we need statistical tools rather than majority rule heuristics for inferring the species tree based on gene trees.

Current methods for inferring species trees from gene trees in this setting can be divided into topology-based and genealogy-based methods, in which the input for a reconstruction algorithm accepts either gene tree topologies or genealogies, i.e., gene trees with branch lengths (coalescence times).  Topology-based methods include Minimize Deep Coalescence (MDC) \cite{maddison2006,than2009}, STAR \cite{LiuEtAl09:systbiol}, STELLS \cite{wu2011}, rooted triple consensus \cite{EwingEtAl08} and other consensus and supertree methods \cite{DegnanEtAl09,WangDegnan2011}.  Genealogy-based methods include Bayesian and likelihood methods such as BEST, *BEAST, and STEM \cite{heled2010,kubatko2009,Liu2007} and clustering and distance-based methods \cite{Liu2011,LiuEtAl10:jmathbiol,LiuEtAl09:systbiol,MosselAndRoch10}.  Possible pros and cons of the two approaches are that topology-based methods can be computationally faster and less sensitive to errors in estimating gene trees (and gene tree branch lengths) from sequence data \cite{huangKubatkoKnowles2010}, while methods that use coalescence times, particularly using Bayesian modelling, can be the most accurate when model assumptions are correct \cite{liuReview2009}.

Another possibility that has been so far unexplored in methods for inferring species trees from gene trees is to use {\it ranked} gene trees, in which the temporal order of the nodes of the gene tree (the coalescence times) is used, but not the continuous-valued branch lengths.  This approach might therefore be intermediate between purely topology-based methods and genealogy-based methods.  By preserving more of the temporal information in the gene tree nodes, the hope is to develop methods that are more powerful than purely topology-based methods and that are still computationally efficient and robust to errors in estimating gene trees and gene tree branch lengths from sequence data.

\begin{table*}\label{Table}
\begin{center}\caption{Notation used in the paper}
\begin{tabular}{l l}
\hline \hline Symbol & meaning\\
\hline $\cT$ & species tree with real-valued divergence times\\
$\sT$ & ranked gene tree (real-valued coalescence times not specified)\\
$n$ & the number of leaves of $\cT$ and $\sT$\\
$s_i$ & speciation times, with $s_1 > \cdots > s_{n-1}$, let $s_0 = \infty$\\
$\tau_i$ & intervals between speciation times, $\tau_i = [s_i, s_{i-1})$\\
$\ell_i$ & the number of gene tree lineages at time $s_i$\\
$m_i$ & the number of coalescence events in interval $\tau_i$\\
$\sT_{i,\ell_i}$ & the ranked gene tree observed from time 0 to time $s_i$\\
$g_i$ & the minimum number of gene tree lineages at time $s_i$\\
$y_{i,z}$ & population $z$ in interval $\tau_i$  in beaded tree\\
$u_i$ & internal node (coalescence) with rank $i$ in the gene tree, $u_1$ is most ancient, \\ 
&$u_{n-1}$ is the most recent\\
$k_{i,j,z}$ & the number of lineages available for coalescence in population $y_{i,z}$ just after the $j$th coalescence\\
& (considered forward in time) in interval $\tau_i$; $k_{i,0,z}$ is the number of lineages ``exiting" at time $s_{i-1}$\\
%$b_{i,j}$ & population in which the $j$-th coalescent event (forward in time) in interval $\tau_i$ occurs \\ we do not use this variable any more!
%$\delta(y_{i,z})$, $\delta(u)$ & the number of leaves descended from a node of the species tree or gene tree\\
\james{$\delta(y), \delta(u)$} & \james{the set of leaves descended from a node of the species tree or gene tree, respectively}\\
$\text{lca}(u)$ & for a node $u$ of the gene tree, the node $y$ of the species tree with largest rank such that  $\delta(u) \subset \delta(y)$\\
$\tau(y)$ & \tanja{for a node $y$ with rank $i$ on the species tree, we denote $\tau(y) = \tau_i$ (the interval immediately above $y$)} \\
$\lambda_{i,j}$ & the overall coalescence rate in interval $\tau_i$ immediately preceding (backwards in time) \\
& the $j$th coalescence\\
$h_k^1$ & number of sequences of coalescences above the root of the species tree starting with $k$ lineages\\
$f_i$ & the joint density of coalescence times in interval $\tau_i$\\
\hline
\end{tabular}
\end{center}
\end{table*}

In \cite{Degnan2012}, a first step toward developing methods that use ranked gene trees for inferring species trees was taken by providing formulae to calculate the probability of a ranked gene tree given a species tree.  \james{The previous work, however, was based on an exponential enumeration of what were called {\em ranked coalescent histories} and did not provide an algorithm for computing some of the key terms in the probability of individual ranked histories.} \tanja{
In this paper, we improve this previous (computationally inefficient) approach, by providing
 a method for computing probabilities of ranked gene trees given species trees which is polynomial in the number of leaves using a dynamic programming approach.}  
 
 Methods for computing probabilities of ranked gene trees efficiently may also be of interest in the context of computing probabilities of unranked gene trees, particularly because no polynomial time algorithm has been found for calculating the probability of a gene tree topology given a species tree under the multispecies coalescent \cite{DegnanAndSalter05,Rosenberg07:jcb,ThanEtAl07,wu2011}.  The probability of an unranked gene tree topology can be obtained by summing over all ranked gene tree topologies with the same topology.   Thus, for unranked gene trees with particular shapes where the number of rankings increases in polynomial time, using ranked gene trees can potentially increase the speed of computing probabilities of unranked gene trees as well. \james{We note that a completely unbalanced gene tree has only one ranking, while the number of rankings can be exponential in the number of leaves when gene trees become more balanced.  Thus, our approach for calculating unranked gene tree probabilities will be most useful for less balanced ranked gene trees.
%JAMES: is there an analytic result, saying something like: x\% of all trees with $n$ leaves have an exponential number of rankings?
%REPLY: Not that I know of.
}

The bulk of the paper consists of the derivation of the polynomial time method for computing ranked gene tree probabilities.  \james{The algorithm is summarized in section \ref{SecAlg}.} This is followed by a discussion of applications to computing probabilities of unranked gene tree topologies and to inferring ranked species trees under maximum likelihood and a modification to the MDC criterion.

%Designing such statistical tools requires an understanding of the gene tree topology distribution induced by a species tree. Fast algorithms to determine the probability of a gene tree given a species tree will allow us to investigate the gene tree topology distribution and thus may potentially lead to a statistical framework of inferring the species tree based on gene tree topologies. 
%In this paper we use the concept of ranked gene tree topologies: ranked gene tree topologies are gene tree topologies with the internal vertices being ordered. We provide a polynomial time algorithm for calculating the probability of a ranked gene tree topology,  given a species tree. Our algorithm is a dynamic programming approach.

\section{Calculating the probability of a ranked gene tree topology}
In the following, we will derive the probability of a ranked gene tree topology given a species tree, $\bP[\sT \,|\,\cT] $.  
Equations (\ref{EqnResult},\ref{EqnHelp},\ref{EqnRec},\ref{EqnGcondG},\ref{E:gi},\ref{E:kijz2}) allow the calculation of $\bP[\sT \,|\,\cT] $ in time $O(n^5)$.
The model giving rise to the gene tree is the multi-species coalescent with constant population sizes \cite{DegnanAndRosenberg09}. Each species consists of a population of constant size where lineages merge according to the coalescent. Thus, lineages from two different species may coalesce any time previous to  the split of the two species.

We begin with some notation, which is also summarized in Table 1.  Let time be $0$ today and increasing going into the past.
Let $\cT$ be a species tree with $n$ species, and thus $n-1$ speciation events (denoted by $1,\ldots,n-1$) occurring at times $s_1> \cdots > s_{n-1}$. 
Denote the interval between speciation event $i-1$ and speciation event $i$ by $\tau_i$, see Figure 1.

Let $\sT$ be a ranked gene tree topology.   It is convenient to use the same labels for the leaves of $\sT$ and of $\cT$.  This is a slight abuse of notation, as leaf $A$ of $\cT$ refers to a population (or species), and $A$ of $\sT$ refers to a gene sampled from population $A$.
We denote the nodes of $\sT$ (which are coalescence events) by $u_1, \dots, u_{n-1}$, where node $u_j$ has rank $j$, and where higher rank indicates a more recent coalescence.  A ranked tree topology can be notated similarly to Newick notation, putting the rank as a subscript for each node, see also Figure 1.

Let %$\bP_i[\sT,\ell_i\,|\,\cT]$ 
$\sT_{i,\ell_i}$ be \tanja{ part of} a ranked gene tree evolving on a species tree between time $s_i$ and \tanja {time $0$ (i.e. the present)}.  $\sT_{i,\ell_i}$ consists of $\ell_i$ \tanja{ gene tree} lineages at speciation time $s_{i}$ and the coalescent history of  $\sT_{i,\ell_i}$  in time interval $(0,s_i)$ is consistent with the ranked gene tree $\sT$. 
Let $g_i$ be the minimum number of lineages required in the ranked gene tree at time $s_i$ such that $\sT$ can be embedded into the species tree $\cT$.
Note that $n\geq \ell_i \geq g_i >i$. Next we provide a dynamic programming approach for calculating the probability of a ranked gene tree given a species tree. An efficient way to determine the required quantities $g_1,\ldots,g_{n-1}$ is provided in Section \ref{SecGK}.

\tanja{Essentially, in our approach, we traverse the intervals between speciation events going back in time, $\tau_{n-1},\ldots,\tau_2$ (formalized in Theorem \ref{Thm2}), and calculate the probability of the appropriate coalescent events occuring in interval $\tau_i$ based on how many coalescent events happened in the later intervals $\tau_{i+1},\ldots, \tau_{n-1}$ (Theorem \ref{Thm3}). Finally with Theorem \ref{pT}, we account for the most ancetral time interval $\tau_1$.}

\begin{theo} \label{pT}
The probability of a ranked gene tree given a species tree is,
\begin{equation}
\bP[\sT \,|\,\cT] = \sum_{\ell_1=g_1}^n \bP[\sT_{1,\ell_1} \,|\,\cT] / H_{\ell_1} \label{EqnResult}
\end{equation}
where 
\begin{equation}\label{EqnHelp}
H_{\ell_1} = \ell_1!(\ell_1-1)!/2^{\ell_1-1}
\end{equation}
 is the probability for the coalescences above the root appearing in the right order \cite{edwards1970}.
\end{theo}

\tanja{For precalculated $ \bP[\sT_{1,\ell_1}| \cT]$ ($\ell_1=2,\ldots,n$) the complexity of calculating $\bP[\sT \,|\,\cT] $ is thus $O(n)$.
Next, we will provide a recursive way to calculate $ \bP[\sT_{1,\ell_1}| \cT]$ for $\ell_1 =2, \ldots, n$ in polynomial time, thus $\bP[\sT \,|\,\cT]$ can be calculated in polynomial time.}
%For doing so, let $\bP_{i,i+1}[\sT,\ell_i\,|\, \ell_{i+1}, \cT]$ be the probability that at time $s_{i}$ we have $\ell_i$ lineages in the gene tree and the coalescent events in the interval $\tau_i$ are consistent with the ranked gene tree $\sT$, given that at speciation time $s_{i+1}$ we have $\ell_{i+1}$ lineages. 

\begin{theo} \label{Thm2}
\tanja{The probability $ \bP[\sT_{i,\ell_i}|\cT]$ can be calculated for all $i$ recursively (with $l_i\geq g_i$),}
\begin{align} \label{EqnRec}
 &\bP[\sT_{i,\ell_i} | \cT] \\
 &=    \underset{\ell_{i+1}=\max (\ell_i,g_{i+1})}{\overset{n}{\sum}}    \hspace{-.5cm} \bP[\sT_{i,\ell_i} | \sT_{i+1,\ell_{i+1}} , \cT]  \bP[\sT_{i+1,\ell_{i+1}} | \cT] \notag %, & \mbox{ if } \ell_i \geq g_i; \\ 0,\, & \mbox{ else  }; 
% \end{cases}
\end{align}
with
$$\bP[\sT_{n-1,n}\,|\,\cT] = 1.$$
The complexity of calculating $ \bP[\sT_{1,\ell_1}\,|\,\cT] $ for $\ell_1 =2, \ldots, n$  is $O(n^3)$, given we know $ \bP[\sT_{i,\ell_i} \,|\, \sT_{i+1,\ell_{i+1}}, \cT] $ for all $i,\ell_i,\ell_{i+1}$.
\end{theo}
\begin{proof}
At the time of the most recent speciation event, $s_{n-1}$, we have $n$ lineages with probability $1$, which is the initial value of the recursion.
Calculating $ \bP[\sT_{i,\ell_i} | \cT]$ for $i<n-1$ can be done in the following way,
\begin{align*}
 &\bP[\sT_{i,\ell_i} | \cT] \\
 &= \sum_{\ell_{i+1}=\max (\ell_i,g_{i+1})}^n   \bP[\sT_{i,\ell_i},  \sT_{i+1,\ell_{i+1}} | \cT]  \notag \\
&=  \sum_{\ell_{i+1}=\max (\ell_i,g_{i+1})}^n    \bP[\sT_{i,\ell_i} | \sT_{i+1,\ell_{i+1}}, \cT  ]  \bP[\sT_{i+1,\ell_{i+1}} | \cT].  \notag
\end{align*}
Suppose $ \bP[\sT_{i,\ell_i} | \sT_{i+1,\ell_{i+1}}, \cT  ] $ is known.
Given we calculated the probability $  \bP[\sT_{i+1,\ell_{i+1}} | \cT]  $ for  $\ell_{i+1} = i+2,\ldots, n$, then calculating $ \bP[\sT_{i,\ell_i} | \cT]$ for $\ell_i = i+1,\ldots,n$ requires $O(\sum_{j=1}^{n-i} j ) = O({n-i+1 \choose 2})$ calculations.
Summing up over $i=1,\ldots,n-1$ yields a complexity of $O(\sum_{i=2}^n {i \choose 2})= O({n+1 \choose 3}) = O(n^3)$.
\end{proof}

It remains to determine $  \bP[\sT_{i-1,\ell_{i-1}} | \sT_{i,\ell_{i}}, \cT  ]$. Note
that during the interval $\tau_i$, we have $i$ branches in the species
tree. 
Let $m_i$ be the number of coalescent events in $\tau_i$, so $m_i=\ell_{i}-\ell_{i-1}$. 
%Assume that coalescent event $j$ in interval $\tau_i$ occurs on branch $b_{i,j}$, $1 \le j \le m_i$. 
%Coalescent event $j$ occurs at time $v_j$ prior to coalescent event $j+1$, 
Let the number of lineages on branch $z$ just after the $j$th
coalescent event (going forward in time) in $\tau_i$ be $k_{i,j,z}$. Calculation of $k_{i,j,z}$ can be done efficiently as shown in Section \ref {SecGK}. 

%For $j=0$, interpret
%$k_{i,j,z}$ to be the number of lineages in branch $z$ at the boundary
%of $\tau_i$ and $\tau_{i-1}$, we get $\sum_{z=1}^{i} k_{i,0,z}=\ell_i$.

\begin{theo}
\label{Thm3}
 We have,
\begin{eqnarray}
 \bP[\sT_{i-1,\ell_{i-1}} | \sT_{i,\ell_{i}}, \cT  ] &=&  \sum_{j=0}^{m_i}   \frac{e^{-\lambda_{i,j} (s_{i-1}-s_i)}}{\prod_{k=0,k\neq j}^{m_i} (\lambda_{i,k}-\lambda_{i,j})} \label{EqnGcondG}
% e^{ \sum_{z=1}^{i} \left( \binom{k_{i,m_i,z}}{2}s_i -  \binom{k_{i,0,z}}{2}s_{i-1}\right)  }  % \prod_{j=1}^{m_i} \binom{k_{i,j,b_{ij}}}{2}
%  g(i,m_i,\ell_{i})
 \end{eqnarray}
 where $\lambda_{i,j} = \sum_{z=1}^{i}  \binom{k_{i,j,z}}{2}$ and $\binom{1}{2}:=0$.
%where $$g(i,m_i,\ell_{i}) = \int_0^{s_{i-1}-s_i} \int_0^{u_{i,1}} \ldots \int_0^{u_{i,m_i-1}} e^{- \sum_{j=1}^{m_i}  \lambda_{i,j} u_{i,j}}  du_{i,m_i},\ldots,   du_{i,1},$$ and $\lambda_{i,j}= k_{i,j,b_{ij}}-1$.
\end{theo}

\begin{proof}
%Let $g_k(s)$ be the probability that $k$ lineages do not coalesce during time $s$. As $k$ lineages coalesce with a rate ${k \choose 2}$, we have 
%$g_k(s) = e^{-{k \choose 2}s}$. 

The density for the coalescence events in interval $\tau_i$ can be obtained by considering the waiting time to the ``next" coalescent event (going backwards in time) as being due to competing exponentials in the different branches, where the coalescence rate within branch $z$ is $\binom{k_{i,j,z}}{2}$.  Thus, the waiting time until the next coalescent event \tanja{has rate} $\lambda_{i,j} = \sum_{z=1}^i \binom{k_{i,j,z}}{2}$.  

We denote the time between the $j$th and $(j+1)$st coalescent event as $v_j$, where $v_0$ is the time between $s_{i-1}$ and the first (least recent) coalescent event in $\tau_i$ and
with $v_{m_i}$ being the time between $s_i$ and coalescent event $m_i$.

The density for the coalescent events in the interval $\tau_{i}$ is \cite{Degnan2012},
   \begin{eqnarray*}
f_i(v_{0}, v_1, \ldots, v_{m_i}) &=&
%  \prod_{j=1}^{m_i}% \binom{k_{i,j,b_{ij}}}{2}  
 e^{-  \sum_{j=0}^{m_i}\sum_{z=1}^{i}  \binom{k_{i,j,z}}{2}v_j}\\
  &=&% \prod_{j=1}^{m_i}% \binom{k_{i,j,b_{ij}}}{2}  
  e^{-  \sum_{j=0}^{m_i}   \lambda_{i,j} v_j}.  %\sum_{z=1}^{i}  \binom{k_{i,j,z}}{2}}
     \end{eqnarray*}
     
   \tanja{  It remains to integrate over $v$, for which we distinguish between case (i) $\lambda_{i,0} = 0$, and case (ii) $\lambda_{i,0} > 0$.}
     
Case (i):   If $\lambda_{i,0} = 0$ (which occurs if $\ell_{i-1} = i$, i.e., all lineages within each population coalesce), then we rewrite $f_i$ as, 
      \begin{equation}\label{E:fi1}
      f_i(v_0,v_1, \dots, v_{m_i}) =   \frac{ \prod_{j=1}^{m_i} \lambda_{i,j} e^{-   \lambda_{i,j} v_j  }}{ \prod_{j=1}^{m_i} \lambda_{i,j} }.
      \end{equation}
      Using the fact that the integral of the numerator of Equation \eqref{E:fi1} is  a hypoexponential distribution based on the sum of $m_i$ exponential random variables \cite{Ross2007} \tanja{(with density functions $\lambda_{i,j} e^{-   \lambda_{i,j} v_j  }$, $j=1,\ldots,m_i$)}, the probability of the coalescent events in the interval is the {\it cumulative distribution function} of the hypoexponential distribution evaluated at $s_{i-1} - s_i=\sum_{j=0}^{m_i} v_i$.  Thus, with $\lambda_{i,j}<\lambda_{i,j+1}$, 
      
      \tanja{
      \begin{align}
       \bP[\sT_{i-1,\ell_{i-1}} &| \sT_{i,\ell_{i}}, \cT  ] \notag\\
       &= \frac{1}{\prod_{j=1}^{m_i}\lambda_{i,j}} - \sum_{j=1}^{m_i} \frac{e^{-\lambda_{i,j}(s_{i-1}-s_i)}}{\lambda_{i,j}\prod_{k=1,k\neq j}^{m_i} (\lambda_{i,k}-\lambda_{i,j})}\notag \\
       &=  \frac{1}{\prod_{j=1}^{m_i}\lambda_{i,j}} + \sum_{j=1}^{m_i} \frac{e^{-\lambda_{i,j}(s_{i-1}-s_i)}}{\prod_{k=0,k\neq j}^{m_i} (\lambda_{i,k}-\lambda_{i,j})}\notag \\
 %      &= \frac{1}{\prod_{j=1}^{m_i} \lambda_{i,j}} + \sum_{j=1}^{m_i}\frac{e^{-\lambda_{i,j}(s_{i-1}-s_i)}}{\prod_{k=0,k\neq j}^{m_i} (\lambda_{i,k}-\lambda_{i,j})}\\
      &= \sum_{j=0}^{m_i}\frac{e^{-\lambda_{i,j}(s_{i-1}-s_i)}}{\prod_{k=0,k\neq j}^{m_i} (\lambda_{i,k}-\lambda_{i,j})}.\label{E:lambdai0}
       \end{align}
       where the second line follows because $-\lambda_{i,j} = \lambda_{i,0}-\lambda_{i,j}$.
      }

 Case (ii):   If $\lambda_{i,0} > 0$, then we rewrite $f_i$ as,
   \begin{eqnarray}
f_i(v_{0}, v_1, \ldots, v_{m_i}) &=&
  \frac{ \prod_{j=0}^{m_i} \lambda_{i,j} e^{-   \lambda_{i,j} v_j  }}{ \prod_{j=0}^{m_i} \lambda_{i,j} }\label{E:fi2}
      \end{eqnarray}
     
     For integrating $f_i$, we use the fact that the integral of the numerator in Equation \eqref{E:fi2} is the convolution of $m_i+1$ exponential random variables with parameters $\lambda_{i,0},\ldots,\lambda_{i,m_i}$, which is the hypoexponential distribution.
      Now, since $\lambda_{i,j}<\lambda_{i,j+1}$, we observe, using the {\it probability density function} of the hypoexponential distribution,
       \begin{align*}
 \bP[\sT_{i-1,\ell_{i-1}} &| \sT_{i,\ell_{i}}, \cT  ] \\
 &= \int_v f_i(v_{0}, v_1, \ldots, v_{m_i})\; dv \\
% &=  \frac{\sum_{j=0}^{m_i} \lambda_{i,j} e^{-\lambda_{i,j} (s_{i-1}-s_i)} \prod_{k=0,k\neq j}^{m_i} \frac{\lambda_{i,k}}{\lambda_{i,k}-\lambda_{i,j}}}{ \prod_{j=0}^{m_i} \lambda_{i,j} }\\
 &=  \sum_{j=0}^{m_i}   \frac{e^{-\lambda_{i,j} (s_{i-1}-s_i)}}{ \prod_{k=0,k\neq j}^{m_i} (\lambda_{i,k}-\lambda_{i,j})},
      \end{align*}
      which is the same expression as for the $\lambda_{i,0}=0$ case \eqref{E:lambdai0}.
     \tanja{ Note that for case (i) we made use of the cumulative distribution function of the hypoexponential distribution, while for case (ii) we made use of the density  function of the hypoexponential distribution. Both cases yield the same final expression for $\bP[\sT_{i-1,\ell_{i-1}} | \sT_{i,\ell_{i}}, \cT  ] $, which establishes the proof.}
   \end{proof}
      \begin{coro} \label{CorComplexity}
  The probabilities $ \bP[\sT_{i-1,\ell_{i-1}} | \sT_{i,\ell_{i}}, \cT  ] $  for all possible $i$, $m_i$ and $\ell_i$ (recall that $m_i=\ell_i-\ell_{i-1}$) are calculated in $O(n^5)$, given all $\lambda_{i,j}$.
      \end{coro}
\begin{proof}
For a fixed  $i$, $m_i$ and $\ell_i$, we require $O(m_i^2)$ calculations to evaluate  $ \bP[\sT_{i-1,\ell_{i-1}} | \sT_{i,\ell_{i}}, \cT]$.  We need to determine $ \bP[\sT_{i-1,\ell_{i-1}} | \sT_{i,\ell_{i}}, \cT  ] $  for all possible $i$, $m_i$ and $\ell_i$. First, we observe that  $i \leq \ell_{i-1}\leq n$, and thus for a fixed $\ell_i$, we have,
$0\leq m_i\leq \ell_i-i$.
Second, $i<\ell_i\leq n$. And third, $2 \leq i \leq n-1$.
%Between time $0$ and $s_i$, $n-\ell_i$ coalescent events occurred. Therefore, $m_i$ can be between $0$ and $(n-i)-(n-\ell_i)=\ell_i-i$. 
Thus, the number of calculations needed to calculate $ \bP[\sT_{i-1,\ell_{i-1}} | \sT_{i,\ell_{i}}, \cT  ]$ for all possible $i$, $m_i$ and $\ell_i$ is,
\begin{align*}
O\left(\sum_{i=2}^{n-1} \sum_{\ell_i=i+1}^n \sum_{m_i=0}^{\ell_i-i} m_i^2\right) &= O\left(\sum_{i=2}^{n-1} \sum_{\ell_i=i+1}^n (\ell_i-i)^3\right) \\
&=  O\left(\sum_{i=2}^{n-1} (n-i)^4\right) \\
&= O\left(n^5\right).
\end{align*}
\end{proof}
    \begin{coro} \label{CorComplexity2}
  The quantities $\lambda_{i,j}$ can be calculated  for all possible $i$, $m_i$, $\ell_i$ and $j$ in $O(n^5)$, given all $k_{i,j,z}$.
        \end{coro}
\begin{proof}
For a fixed  $i$, $m_i$, $\ell_i$ and $j$, we require $O(i)$ calculations to evaluate  $\lambda_{i,j}$.  As $j=0,\ldots,m_i$, with the same arguments as in Corollary \ref{CorComplexity}, we obtain,
\begin{align*}
O\left(\sum_{i=2}^{n-1} \sum_{\ell_i=i+1}^n \sum_{m_i=0}^{\ell_i-i} \sum_{j=0}^{m_i} i\right) 
&= O\left(\sum_{i=2}^{n-1} i \sum_{\ell_i=i+1}^n \sum_{m_i=0}^{\ell_i-i} m_i \right)  \\
&=   O\left(\sum_{i=2}^{n-1} i \sum_{\ell_i=i+1}^n (l_i-1)^2 \right)  \\
&=   O\left(\sum_{i=2}^{n-1} i (n-i)^3 \right)  \\
&= O\left(n^5\right).
\end{align*}
\end{proof}

   We note that the terms  $\bP[\sT_{i-1,\ell_{i-1}} | \sT_{i,\ell_{i}}, \cT  ] $ are analogous to the functions $g_{i,j}$ defined in  \cite{tavare1984,wakeley2008}, which give the probability that $i$ lineages coalesce into $j$ within time $t$ in a single population and are used extensively in computing probabilities related to unranked gene trees \cite{DegnanAndSalter05,PamiloAndNei88,Rosenberg02,wu2011}.  In particular, if only one population, say $z^{*}$, has coalescence events, then we have
   \begin{align*}
   &\bP[\sT_{i-1,\ell_{i-1}} | \sT_{i,\ell_{i}}, \cT  ] \\
   &= \frac{g_{\ell_{i+1},\ell_i}(s_i-s_{i+1})\prod_{z \ne z^{*}} g_{k_{i,0,z},k_{i,0,z}}(s_i-s_{i+1})}{\prod_{k=1}^{\ell_{i+1}-\ell_i} \binom{\ell_{i+1} - k+1}{2}},
   \end{align*} 
   a product of $g_{i,j}$ functions with the denominator counting the number of sequences in which $m_i$ coalescences could have occurred.
   The terms $ \bP[\sT_{i-1,\ell_{i-1}} | \sT_{i,\ell_{i}}, \cT  ] $ allow for the coalescences to occur in separate populations, however, and are constrained by the ranking of the gene tree.  For example, in interval $\tau_3$ of Figure 1c, there are two coalescences which occur in different populations.  If the ranking of the gene tree were not important, the branches could be considered independent, and the probability of this event would be $g_{2,1}(s_2-s_3)g_{2,1}(s_2-s_3)$.  However, the gene tree ranking constrains the coalescence of $A$ and $B$ to be less recent than that of $D$ and $E$, so the probability for events in this interval is, 
   $$\bP [\sT_{3,2} | \sT_{4,3},\cT] = [g_{2,1}(s_2-s_3)]^2/2.$$
 \james{We illustrate that we get the same result from Theorem \ref{Thm3}: there are two coalescence events in interval $\tau_3$, so we use $j=0,1,2$, and calculate
 \begin{align*}
 &\lambda_{3,0} =   \binom{1}{2} + \binom{1}{2} + \binom{1}{2} = 0, \\
 &\lambda_{3,1} = \binom{2}{2} + \binom{1}{2} + \binom{1}{2} = 1,\\ 
 &\lambda_{3,2} = \binom{2}{2} + \binom{1}{2} + \binom{2}{2} = 2.\\
 \end{align*}
 Thus, Equation \eqref{EqnGcondG} from Theorem \ref{Thm3} evaluates to
 \begin{align*}
 &\frac{e^{-0(s_2-s_3)}}{(2-0)(1-0)} + \frac{e^{-1(s_2-s_3)}}{(0-1)(2-1)} + \frac{e^{-2(s_2-s_3)}}{(0-2)(1-2)}\\
 &=\;\;\;\; \frac{1}{2} -e^{-(s_2-s_3)} + \frac{1}{2}e^{-2(s_2-s_3)}\\
 &= \;\;\;\; \frac{1}{2}\left(1-e^{-(s_2-s_3)}\right)^2\\
 &= \;\;\;\;  [g_{2,1}(s_2-s_3)]^2/2.
 \end{align*}
 }

\begin{rem}
The probability of a gene tree topology is the sum of the probabilities of each ranked gene tree with the given topology.
A given tree topology has $(n-1)!/{\prod_{i=1}^{n-1} (c_i-1)}$ rankings, where $c_i$ is the number of descendant leaves of interior vertex $i$. A proof can be found in \cite{Steel2003}. 
For a completely balanced tree on $n=2^k$ leaves, the number of rankings grows faster than polynomial:
the numerator can be approximated by,
$$n! \approx \sqrt{2 \pi n} (n/e)^n,$$
and the denominator can be approximated by,
$$\prod_{i=1}^{n-1} (c_i-1)=\prod_{i=1}^k (2^i -1)n/2^i \approx n^k=n^{\log_2 n},$$
showing that the ratio grows faster than polynomial in $n$.
\end{rem}

\subsection{Calculation of $g_i$ and  $k_{i,j,z}$} \label{SecGK}
\subsubsection*{Calculation of $g_i$}
If $\cT$ and $\sT$ have the same ranked topology, then $g_i=i+1$.  
In general, to compute $g_i$, we let $\text{lca}(u_j)$ be the {\it least common ancestor} node on the species tree for a node $u_j$ on the ranked gene tree -- i.e., the node with the largest rank on the species tree which is ancestral to all species represented in $u_j$.   For a node $y$ on the species tree, let $\tau(y)$ be the interval immediately above $y$.  For example, in Figure 1c, $\tau(\text{lca}(u_4)) = \tau_3$ where $u_4$ is the gene tree node with rank 4 --- the node ancestral to D and E only.  We then express $g_i$ as
\begin{equation}\label{E:gi}
%\tanja{ g_i = n - \sum_{j=i+1}^{n-1} I((\tau(\text{lca}(u_j))>\tau_i),\ldots,(\tau(\text{lca}(u_{n-1}))>\tau_i))}
\james{ g_i = n - \sum_{j=i+1}^{n-1} \prod_{k=j}^{n-1} I(\tau(\text{lca}(u_k))>\tau_i)}
\end{equation}
where $\tau_j<\tau_i$ iff $j<i$, and 
where $I(\cdot)$ is an indicator function taking the value 1 if the condition holds and otherwise 0.  \james{Assuming each lca() operation is $O(1)$ \cite{harel1984,schieber1988}), preprocessing allows all lca terms to be computed in $O(n)$ time.  Similarly all needed products and the sum in Equation \eqref{E:gi} can each then be computed in $O(n)$ time.  }  Thus,
calculating $g_1,\ldots,g_{n-1}$ can be done in $O(n)$ time.

\subsubsection*{Calculation of $k_{i,j,z}$}

\begin{figure}
 \begin{center}
 \begin{picture}(600.0, 220.0)(0,0)
\put(-20,0){\includegraphics[width=.49\textwidth]{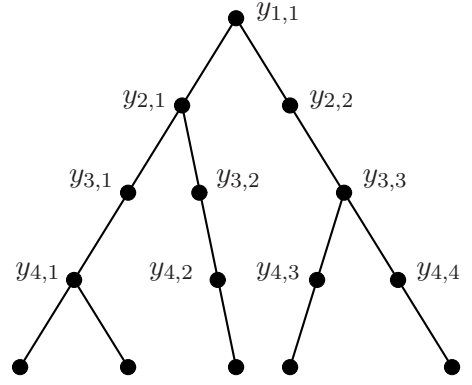}}
\put(30.00,50.00){\fontsize{11.23}{14.07}\selectfont   \makebox(72.0, 100.0)[c]{$y_{4,1}$\strut}}
\put(80.00,50.00){\fontsize{11.23}{14.07}\selectfont   \makebox(72.0, 100.0)[c]{$y_{4,2}$\strut}}
\put(120.00,50.00){\fontsize{11.23}{14.07}\selectfont   \makebox(72.0, 100.0)[c]{$y_{4,3}$\strut}}
\put(180.00,50.00){\fontsize{11.23}{14.07}\selectfont   \makebox(72.0, 100.0)[c]{$y_{4,4}$\strut}}
\put(50.00,85.00){\fontsize{11.23}{14.07}\selectfont   \makebox(72.0, 100.0)[c]{$y_{3,1}$\strut}}
\put(105.00,85.00){\fontsize{11.23}{14.07}\selectfont   \makebox(72.0, 100.0)[c]{$y_{3,2}$\strut}}
\put(160.00,85.00){\fontsize{11.23}{14.07}\selectfont   \makebox(72.0, 100.0)[c]{$y_{3,3}$\strut}}
\put(70.00,115.00){\fontsize{11.23}{14.07}\selectfont   \makebox(72.0, 100.0)[c]{$y_{2,1}$\strut}}
\put(140.00,115.00){\fontsize{11.23}{14.07}\selectfont   \makebox(72.0, 100.0)[c]{$y_{2,2}$\strut}}
\put(120.00,148.00){\fontsize{11.23}{14.07}\selectfont   \makebox(72.0, 100.0)[c]{$y_{1,1}$\strut}}
\end{picture}
\vspace{-1in}
\caption{The beaded version of the species tree topology in Figure 1a--d.}
\end{center}
\end{figure}

%We let $\delta(y)$ be the set of descendant leaves of a node $y$ on a species tree $\cT$ (similarly, $\delta(u)$ is the set of descendant leaves of a node $u$ on a gene tree $\sT$). 
  We let $y_{i,j}$ be the $j$th population (read left to right) in interval $\tau_i$ (equivalently, the $j$th branch or $j$th node subtending the branch).    In order to label every population before and after a speciation time $s_i$ uniquely, extra nodes can be added to the species tree to form a {\it beaded species tree} (Figure 2), so that there are $i$ nodes at time $s_i$, $i=1, \dots, n-1$.  For each $i  \in \{1, \dots, n-1\}$, there is one node of outdegree 2, and $i-1$ nodes of outdegree 1.  Thus, population $y_{i,j}$ corresponds to a branch (equivalently, a node) in the beaded species tree.  We denote the outdegree of a node $y$ by $outdeg(y)$.

In the remainder of this section, we compute the values $k_{i,j,z}$, i.e. the number of lineages on branch $y_{i,z}$ of the beaded species tree during the interval immediately after the $j$th coalescence event (going forward in time), with $k_{i,0,z}$ being the number of lineages ``exiting" the branch at time $s_{i-1}$.  For example, in Figure 1b, we have

\begin{center}
\begin{tabular}{c c c}
$k_{2,0,1} = 1,$ & $k_{2,1,1} = 2,$ & $k_{2,2,1} = 3,$ \\
$k_{2,0,2} = 1,$ & $k_{2,1,2} = 1,$ & $k_{2,2,2} = 1\phantom{,}$\\
\end{tabular}
\end{center}

The value of $k_{i,j,z}$ depends on the number of lineages entering branch $i$, $\ell_{i}$, as well as the number of lineages exiting the branch, and not just on the number of coalescence events in the interval.  For example, in Figure 1c, $k_{2,0,1} = 1$ and $k_{2,1,1} = 2$, while in Figure 1d, $k_{2,0,1} = 2$ and $k_{2,1,1} = 3$, although the two gene trees have the same ranked topology and $m_2 = 1$ for both cases.

To determine the terms $k_{i,j,z}$ we note that the number of coalescences that have occurred more recently than interval $\tau_i$ is $n-\ell_i$. \james{In a given interval $\tau_i$, we let $z^{(1)}$ and $z^{(2)}$ be the left and right children, respectively, of population $z$ of outdegree 2, and let $z^{(1)}= z^{(2)}$ be the only child of a node $z$ of outdegree 1.}

\tanja{
The number of lineages available to coalesce in population $z$ of interval $\tau_i$ is
\begin{equation}\label{E:kijz1}
k_{i,m_i,z} = \sum_{j=1}^{outdeg(y_{i,z})} k_{i+1,0,z^{(j)}}
\end{equation}
where the $z^{(j)}$ are the daughter populations (one or two) of $z$. Further, $k_{n,0,z}=0$ for all $z$.
Since the beaded species tree has $n^2/2$ nodes, precalculating $outdeg(y_{i,z})$ requires $O(n^2)$.
For $\james{0 \le }\, j< m_i$, we have
\begin{align}\label{E:kijz2}
k_{i,j,z} &= \begin{cases}
k_{i,j+1,z}-1 & \text{$j$th coalescence on branch $z$}\\
k_{i,j+1,z} & \text{otherwise}
\end{cases}
\end{align}
Consequently, determining a particular $k_{i,j,z}$ is $O(1)$.
Thus determining $ k_{i,j,z}$ 
 for all possible $i$, $m_i$ and $\ell_i$ is (see also Corollary \ref{CorComplexity}), 
\begin{align*}
&= O\left(\sum_{i=2}^{n-1} \sum_{\ell_i=i+1}^n \sum_{m_i=0}^{\ell_i-i} \sum_{j=0}^{m_i}O(1)   \right)   \\
&= O\left(n^4\right).
%& O\left(\sum_{i=2}^{n-1} \sum_{\ell_i=i+1}^n \sum_{m_i=0}^{\ell_i-i}  \sum_{j=0}^{m_i} (n-l_i+j)        \right)   \\
%&= O\left(\sum_{i=2}^{n-1} \sum_{\ell_i=i+1}^n \sum_{m_i=0}^{\ell_i-i} m_i^2 +2 (n-l_i)m_i     \right)   \\
%&= O\left(n^5\right).
\end{align*}
Note that taking the sum over all $z$ is not necessary, as in all but one branch the $k_{i,j,z}$ equals the $k_{i,j+1,z}$.
}

\subsection{An algorithm} \label{SecAlg}
\tanja{In summary, we derived an algorithm with runtime $O(n^5)$ for calculating the probability of a ranked gene tree given a species tree on $n$ tips:}
\begin{enumerate}
\item Calculate $g_1,\ldots g_{n-1}$ using Equation (\ref{E:gi}).
\item Calculate $k_{i,j,z}$ (for $i,j = 1,\ldots,n; z=1\ldots i$), using Equations (\ref{E:kijz1}) and (\ref{E:kijz2}).
\item Calculate $\lambda_{i,j} = \sum_{z=1}^i {k_{i,j,z} \choose 2}$ (for $i,j = 1,\ldots,n$).
\item Calculate $\bP[\sT_{i-1,\ell_{i-1}} | \sT_{i,\ell_{i}}, \cT  ] $ (for $i=2,\ldots,n$; $\ell_{i-1}=g_{i-1},\ldots,n$;  $\ell_{i}=g_{i},\ldots,n$), using Theorem \ref{Thm3}.
\item Calculate  $ \bP[\sT_{1,\ell_1}|\cT]$ using Theorem \ref{Thm2}.
\item Calculate $\bP[\sT \,|\,\cT]$ using Theorem \ref{pT}.
\end{enumerate}

\section{Discussion}

In this paper, we provide a polynomial-time algorithm ($O(n^5)$ where $n$ is the number of species) to calculate the probability of a ranked gene tree topology given a species tree, summarized in Section \ref{SecAlg}. We now discuss applying these results to computing probabilities of unranked gene tree topologies and to inferring ranked species trees. 
%Thus, for a gene tree topology with a polynomial number of rankings, the probability of this gene tree can be determined in polynomial time.
%It remains an open challenge to determine the complexity of calculating the probability of a gene tree topology given a species tree, where the number of rankings is exponential.  

\subsection{Computing probabilities of unranked gene tree topologies}

Previous work on computing probabilities of unranked gene tree topologies used the concept of \emph{coalescent histories}, which specify the branches in the species tree in which each node of the gene tree occurs.  An unranked gene tree probability can then be computed by enumerating all coalescent histories and computing the probability of each.  The number of coalescent histories grows at least exponentially when the (unranked) gene tree matches the species tree, making this approach computationally intensive.  Coalescent histories can be enumerated either recursively (e.g., in PHYLONET \cite{than2008} or \cite{Rosenberg07:jcb}) or nonrecursively (COAL \cite{DegnanAndSalter05}). 

A much faster approach using dynamic programming similar to that used in this paper is implemented in STELLS \cite{wu2011}, which conditions on the ancestral configuration in each branch rather than the number of lineages.  Here an ancestral configuration keeps track not only of the number of lineages in a branch in the species tree, but also the particular nodes of the gene tree.  Different ancestral configurations can potentially have the same number of lineages within a population.  Enumerating ancestral configurations turns out to have exponential running time for arbitrarily shaped trees, but the number of ancestral configurations is still much smaller than the number of coalescent histories.  When computing probabilities of ranked gene tree topologies, however, the ranking specifies the sequence of coalescence events, leading to a unique ancestral configuration given the number of lineages in a time interval.  This fortuitously enables probabilities of ranked gene tree topologies to be computed in polynomial time.  

We note that although the number of rankings for a gene tree is not polynomial in the number of leaves in general, the number of rankings can be small for certain tree shapes.  For example, if the gene tree has a {\it caterpillar} shape, in which each internal node has a leaf as a descendant, then there is only one ranking, and thus computing the ranked and unranked gene tree are equivalent.  For a {\it pseudo-caterpillar}, a tree made by replacing the subtree with four leaves of a caterpillar with a balanced tree on four leaves \cite{Rosenberg07:jcb}, there are only two rankings possible, and for a {\it bicaterpillar} \cite{Rosenberg07:jcb}, for which the left subtree is a caterpillar with $n_L$ leaves and the right subtree is a caterpillar with $n-n_L$ leaves, there are $\binom{n-2}{n_L-1}$ rankings.  Thus computing unranked gene tree probabilities by summing ranked gene tree probabilities can be done in polynomial time for some tree shapes.  \james{We note that for the approach used by STELLS, some tree shapes can also be computed in polynomial time, including the cases we mentioned that have a polynomial number of rankings.    An open question is whether there are any classes of unranked gene trees which have a polynomial number of rankings but an exponential number of ancestral configurations, or vice versa.}

\subsection{Inferring species trees from ranked gene trees}
Our fast calculation of the probability of ranked gene tree topologies can be used to determine the maximum likelihood species tree from a collection of known gene trees.
Assume we have observed $N$ ranked gene trees (i.e., $N$ loci). 
Now the maximum likelihood species tree $\cT_{ML}$ (with branch lengths on internal branches) is
\begin{align*}
\cT_{ML} &= \underset{\cT}{\text{argmax }} \bP[\sT_1,\ldots,\sT_N|\cT] 
\end{align*}
\james{where  
\begin{align}\label{E:ML}
 \bP[\sT_1,\ldots,\sT_N|\cT]&= \prod_{k=1}^N\bP[\sT_k|\cT] = \prod_{i=1}^{H_n}  \bP[\sT^{(i)}|\cT]^{n_i}\;\; 
\end{align}
\james{is a multinomial likelihood.}
Here $\bP[\sT_k|\cT]$ can be determined with our polynomial-time  algorithm, we let $\sT^{(i)}$ denote the $i$th ranked topology, and $n_i$ is the number of times ranked topology $i$ is observed, with $\sum_{i=1}^{H_n} n_i = N$.}
Note in particular that the \james{ranked topology of $\cT_{ML}$ might differ from the most frequent ranked gene tree topology\cite{Degnan2012}.}

Our derivation of the ranked gene tree probability also suggests a way to infer a ranked species tree topology from ranked gene tree topologies with a similar flavor as the MDC criterion.  In MDC, for an input gene tree and candidate species tree, the number of  extra lineages (lineages which necessarily fail to coalesce due to topological differences between gene and species trees) on each edge of the species tree is counted.  \james{For MDC}, whether the edge of the species tree is long or short does \james{not} affect the deep coalescence cost.  In working with ranked gene trees, \james{however, we can} keep track of the minimum number of extra lineages within each time interval $\tau_i$.  The total number of extra lineages in this sense is 
\begin{equation}\label{E:mdc}
\sum_{i=1}^{n-1} g_i - (i+1)
\end{equation}
Minimizing \eqref{E:mdc}  as a criterion for the ranked species tree will tend to penalize long edges of the species tree which have multiple lineages persisting through multiple species divergence events.  As an example, in Figure 1b, the gene tree has a MDC cost of 1 since there are two lineages exiting the population immediately ancestral to $A$ and $B$; however the cost according \eqref{E:mdc} is 2 because there are two edges on the beaded version of the species tree (Figure 2) that each have an extra lineage.  In Figure 1c, the gene tree has a MDC cost of 0 for the species tree since it has the matching unranked topology; however, the number of extra lineages from equation \eqref{E:mdc} is 1.  We note that in Figure 1c, interval $\tau_3$, incomplete lineage sorting (and deep coalescence) have not occurred as these concepts are normally used.  To capture the idea that coalescence has nevertheless occurred in a more ancient time interval than allowed, we might refer to the coalescence of $A$ and $B$ in Figure 1c as an ``ancient lineage sorting" event (rather than incomplete lineage sorting event) or an ancient coalescence rather than a deep coalescence.  We could therefore refer to minimizing equation \eqref{E:mdc} as the Minimize Ancient Coalescence (MAC) criterion, which would provide an interesting comparison to the usual topology-based MDC criterion.

\james{In practice, a method of inferring a species tree from ranked gene trees would require estimating the ranked gene trees.  This would require clock-like gene trees, or trees with times estimated for nodes, which can also be inferred under relaxed clock models in BEAST \cite{drummond2007}.   To account for the uncertainty in the gene trees, the counts for different ranked gene trees could be weighted by their posterior probabilities obtained from Bayesian estimation of the gene trees \cite{AllmanDegnanRhodes2011}.      Thus, in equation \eqref{E:ML}, we would let  $n_{ik}$ be the posterior probability of ranked topology $i$ at locus $k$, and use $n_i = \sum_{k=1}^{H_n} n_{ik}$ as the estimated number of times that ranked topology $i$ was observed.  Similarly, for equation \eqref{E:mdc}, the coalescence cost at a locus could be distributed over multiple topologies weighted by their posterior probabilities.}

%In this paper, we made the simplifying assumptions of equal population sizes  and equal generation times across lineages.  Of course in empirical data, these assumptions might be violated, and further work is required to relax inappropriate assumptions. This paper introduces a general methodology for calculating ranked gene tree topologies which can be generalized in the future.

\section*{Acknowledgements}
We thank David Bryant for suggesting the dynamic programming approach to this problem and two anonymous referees for valuable comments, particularly on calculating $g_i$ and $k_{i,j,z}$.
JHD was funded by the New Zealand Marsden fund and by a Sabbatical Fellowship at the National Institute for Mathematical and Biological Synthesis, an Institute sponsored by the National Science Foundation, the U.S. Department of Homeland Security, and the U.S. Department of Agriculture through NSF Award \#EF-0832858, with additional support from The University of Tennessee, Knoxville. TS was funded by the Swiss National Science Foundation.

\begin{footnotesize}
\bibliographystyle{plain}
\bibliography{rankedbib}
\end{footnotesize}

\clearpage

\end{document}